\documentclass[a4paper,11pt]{article}
\usepackage{graphicx,epsfig,picins}
\usepackage{fancyhdr,fancybox}
\usepackage{indentfirst}
\usepackage{titlesec}
\usepackage{verbatim}
\usepackage[sort&compress, numbers]{natbib}
\usepackage{array,dcolumn,tabularx}
\usepackage{setspace}
\usepackage{geometry}
\usepackage{extarrows,chemarrow,xypic} 
\usepackage[small]{caption2}
\usepackage{sectsty}
\usepackage{color}

\usepackage{times}     

\usepackage{latexsym}
\usepackage{amsmath}   
\usepackage{amssymb}   
\usepackage{amsbsy}
\usepackage{amsthm}
\usepackage{amsfonts}
\usepackage{mathrsfs}  
\usepackage{bm}        
\usepackage{relsize}   

\usepackage{hyperref}  

\newcommand{\paperfont}{\fontsize{12pt}{1.3\baselineskip}\selectfont}
\sectionfont{\fontsize{13}{15}\selectfont}
\begin{document}


\theoremstyle{definition}
\makeatletter
\thm@headfont{\bf}
\makeatother
\newtheorem{theorem}{Theorem}[section]
\newtheorem{definition}[theorem]{Definition}
\newtheorem{lemma}[theorem]{Lemma}
\newtheorem{proposition}[theorem]{Proposition}
\newtheorem{corollary}[theorem]{Corollary}
\newtheorem{remark}[theorem]{Remark}
\newtheorem{example}[theorem]{Example}
\newtheorem{notation}[theorem]{Notation}

\lhead{}
\rhead{}
\lfoot{}
\rfoot{}

\renewcommand{\refname}{References}
\renewcommand{\figurename}{Figure}
\renewcommand{\tablename}{Table}
\renewcommand{\proofname}{Proof}

\title{\textbf{Novel dissipative properties for the master equation}}
\author{Liu Hong$^1$,\;\;Chen Jia$^2$,\;\;Yi Zhu$^1$,\;\;Wen-An Yong$^{1,*}$ \\
\footnotesize $^1$Zhou Pei-Yuan Center for Applied Mathematics, Tsinghua University, Beijing 100084, P.R. China \\
\footnotesize $^2$Beijing Computational Science Research Center, Beijing 100094, P.R. China \\
\footnotesize $^*$Correspondence: wayong@mail.tsinghua.edu.cn \\}
\date{}                              
\maketitle                           
\thispagestyle{empty}                

\paperfont

\begin{abstract}
Recent studies have shown that the entropy production rate for the master equation consists of two nonnegative terms: the adiabatic and non-adiabatic parts, where the non-adiabatic part is also known as the dissipation rate of a Boltzmann-Shannon relative entropy. In this paper, we provide some nonzero lower bounds for the relative entropy, the entropy production rate, and its adiabatic and non-adiabatic parts. These nonzero lower bounds not only reveal some novel dissipative properties for general nonequilibrium processes which are much stronger than the second law of thermodynamics, but also impose some new constraints on thermodynamic constitutive relations. Moreover, we also provide a mathematical application of these nonzero lower bounds by studying the long-time behavior of the master equation. Extensions to the Tsallis statistics are also discussed, including the nonzero lower bounds for the Tsallis-type relative entropy and its dissipation rate. \\

\end{abstract}

\section{Introduction}
The master equation describes the dynamics of the probability distribution for a Markov jump process and has been widely applied to physics, chemistry, biology, and many other scientific fields \cite{reichl1980modern}. It provides an effective way to model various stochastic systems such as random walks, birth-death processes, Lindblad equations \cite{van1992stochastic}, chemical reaction systems \cite{kurtz1972relationship}, single-molecule enzyme kinetics \cite{ge2012stochastic}, and so on.

In recent years, significant progresses have been made in the field of mesoscopic stochastic thermodynamics \cite{jarzynski2011equalities, seifert2012stochastic, van2015ensemble}. The dynamic foundation of this field turns out to be Markov processes and the master equation plays a fundamental role because any Markov process can generally be approximated by a Markov jump process. In the theory of stochastic thermodynamics, an equilibrium state is defined as a stationary process with detailed balance and the deviation of a system from equilibrium can be  characterized by the concept of entropy production rate \cite{jiang2004mathematical, kjelstrup2008non, zhang2012stochastic}. Motivated by the classical theory of nonequilibrium thermodynamics \cite{de1962non}, the entropy production rate can generally be represented as a bilinear form of the thermodynamics fluxes and forces.
When an open system is driven by a sustained energy supply from the environment, it can approach a nonequilibrium steady state (NESS) as its long-time behavior, with a positive entropy production rate.

It has been shown recently that the entropy production rate $e_p$ for the master equation can be decomposed as the sum of two nonnegative terms: $e_p = e_p^{(ad)}+e_p^{(na)}$, where $e_p^{(ad)}$ (resp. $e_p^{(na)}$) is the adiabatic (resp. non-adiabatic) entropy production rate \cite{ge2010physical, esposito2010three}. The adiabatic part is also known as housekeeping heat \cite{oono1998steady, hatano2001steady} and the non-adiabatic part is also referred to as the dissipation rate $-dF/dt$ of the free energy $F$ \cite{ge2010physical}. Moreover, the free energy $F$ can be represented as a Boltzmann-Shannon relative entropy between the probability distribution $p = (p_i)$ of the system and the steady-state distribution $\mu = (\mu_i)$:
\begin{equation*}
F = k_BT\sum_ip_i\log\frac{p_i}{\mu_i},
\end{equation*}
where $k_B$ is the Boltzmann constant and $T$ is the temperature. The above decomposition is important because it provides a strengthened version of the second law of thermodynamics.

In this work, we shall further explore this strengthened version of the second law. Inspired by the entropy-dissipation principle in nonequilibrium thermodynamics proposed in \cite{yong2004entropy}, we provide nonzero lower bounds for the free energy, the entropy production rate, and its adiabatic and non-adiabatic parts. These nonzero lower bounds reveal new dissipative properties of general nonequilibrium processes which are even much stronger than the strengthened version of the second law. They also reinforce the previous findings that the irreversibility of nonequilibrium processes has two different mechanisms: the deviation from steady state and the breaking of detailed balance \cite{ge2010physical, esposito2010three}.
Moreover, we establish similar conclusions for the Tsallis statistics by providing nonzero lower bounds for the Tsallis-type relative entropy and its dissipation rate. Finally, we elaborate the significance of these nonzero lower bounds from both the physical and mathematical perspectives. It will be seen that these nonzero lower bounds not only impose some new restrictions on thermodynamic constitutive relations, but also provide a simple way to study the long-time behavior of the master equation.

This paper is organized as follows. In Section 2, we introduce some preliminaries on the master equation. Section 3 contains our main results. In Section 4, we present some discussions on the physical significance of the nonzero lower bounds. Section 5 is devoted to a new and purely analytic proof about the long-time behavior of the master equation.

\section{Preliminaries}
Consider a molecular system modeled by a Markov jump process \cite{norris1998markov} with a finite number of states $1,2,\cdots,N$ and transition rate matrix $Q = (q_{ij})$, where $q_{ij}$ with $i\neq j$ denotes the transition rate from state $i$ to $j$ and $q_{ii} = -\sum_{j\neq i}q_{ij}$. Let $p(t) = (p_1(t),p_2(t),\cdots,p_N(t))$ denote the probability distribution of the system at time $t$. Then the dynamics of the probability distribution $p = p(t)$ is governed by the master equation \cite{reichl1980modern}
\begin{equation*}
\frac{dp}{dt} = pQ.
\end{equation*}
In components, the master equation can be written as
\begin{equation}\label{master}
\frac{dp_i}{dt} = \sum_jp_jq_{ji} = \sum_{j\neq i}(p_jq_{ji}-p_iq_{ij}),\;\;\;i = 1,2,\cdots,N.
\end{equation}

We assume that the system is irreducible, namely, for each pair of states $i\neq j$, there is a sequence of states $i_1,i_2,\cdots,i_m$, such that $i_1=i$, $i_m=j$, and
\begin{equation*}
q_{i_1i_2}q_{i_2i_3}\cdots q_{i_{m-1}i_m}>0.
\end{equation*}
Under this irreducible condition, it is well-known that the components of $p(t)$ are all positive for any $t>0$. Furthermore, it can be proved that the system has a unique steady-state distribution $\mu = (\mu_1,\mu_2,\cdots,\mu_N)$ satisfying $\mu Q = 0$ and the components of $\mu$ are all positive \cite{norris1998markov}. In general, the steady state of the system can be classified into equilibrium and nonequilibrium ones. In an equilibrium state, the detailed balance condition $\mu_iq_{ij} = \mu_jq_{ji}$ holds for any pair of states $i$ and $j$. In an NESS, however, the detailed balance condition is broken and the system is externally driven with concomitant entropy production.

Throughout this paper, we set the Boltzmann constant $k_B = 1$ and assume the temperature $T = 1$ for simplicity. Recall that the total entropy production rate $e_p$ of the system is defined as
\begin{equation}\label{ep}
e_p = \frac{1}{2}\sum_{i,j}(p_iq_{ij}-p_jq_{ji})\log\frac{p_iq_{ij}}{p_jq_{ji}} ,
\end{equation}
where $0/0$ is understood to be $1$. In recent years, it has been shown that $e_p$ can be decomposed as the sum of two nonnegative terms: $e_p = e_p^{(ad)}+e_p^{(na)}$ \cite{ge2010physical, esposito2010three}. Here the adiabatic part $e_p^{(ad)}$, also known as housekeeping heat, can be expressed as
\begin{equation}\label{epad}
e_p^{(ad)} = \frac{1}{2}\sum_{i,j}(p_iq_{ij}-p_jq_{ji})\log\frac{\mu_iq_{ij}}{\mu_jq_{ji}} \geq 0
\end{equation}
and the non-adiabatic part $e_p^{(na)}$ can be expressed as
\begin{equation}\label{fd}
e_p^{(na)} = \frac{1}{2}\sum_{i,j}(p_iq_{ij}-p_jq_{ji})\log\frac{p_i\mu_j}{p_j\mu_i} \geq 0.
\end{equation}

To see the nature of $e_p^{(na)}$, we recall the following definition.
\begin{definition}
Let $u = (u_1,u_2,\cdots,u_N)$ and $v = (v_1,v_2,\cdots,v_N)$ be two probability distributions. Then the Boltzmann-Shannon relative entropy (Kullback-Leibler divergence) between $u$ and $v$ is defined as
\begin{equation*}
D(u\parallel v) = \sum_iu_i\log\frac{u_i}{v_i}.
\end{equation*}
\end{definition}

It was shown in \cite{ge2010physical} that, even if the system is away from equilibrium, the free energy $F$ can still be introduced and can be expressed as the Boltzmann-Shannon relative entropy between the probability distribution $p$ and the steady-state distribution $\mu$:
\begin{equation}\label{F}
F = D(p\parallel\mu) = \sum_ip_i\log\frac{p_i}{\mu_i}.
\end{equation}
It is easy to check that the non-adiabatic entropy production rate $e_p^{(na)}$ is exactly the dissipation rate $f_d = -dF/dt$ for the free energy $F$:
\begin{equation}\label{fd}
e_p^{(na)} = f_d = -\frac{dF}{dt} \geq 0.
\end{equation}

The entropy production rate $e_p$ is always nonnegative and this is an equivalent statement of the second law of thermodynamics. However, the above discussion shows that $e_p$ can be further decomposed as the sum of two nonnegative terms, $e_p^{(ad)}$ and $e_p^{(na)}$. Therefore, this decomposition can be viewed as a strengthened version of the second law. In what follows, we shall further explore this strengthened version of the second law by providing some nonzero lower bounds for the free energy, the total entropy production rate, and its adiabatic and non-adiabatic parts.

\section{Results}

\subsection{Nonzero lower bound for the non-adiabatic entropy production rate}\label{free}
We start with the following inequality, which will be used frequently in the sequel.
\begin{lemma}\label{basic}
Let $K$ be a positive number. Then for any $0\leq x,y\leq K$, we have
\begin{equation*}
y\log\frac{y}{x} \geq \frac{1}{2K}(y-x)^2+(y-x).
\end{equation*}
\end{lemma}

\begin{proof}
For any $0<x\leq K$, set $f(x)=x\log x-x$. Then we have $f'(x)=\log x$ and $f''(x)=1/x\geq 1/K$. By the mean-value theorem, for any $0< x,y\leq K$, there exists $\theta$ between $x$ and $y$ such that
\begin{equation*}
f(y)-f(x) = f'(x)(y-x)+\frac{1}{2}f''(\theta)(y-x)^2.
\end{equation*}
Moreover, it can be seen that
\begin{equation*}
f(y)-f(x)-f'(x)(y-x) = y\log\frac{y}{x}-(y-x).
\end{equation*}
This shows that for any $0< x,y\leq K$,
\begin{equation*}
y\log\frac{y}{x}-(y-x) = \frac{1}{2}f''(\theta)(y-x)^2 \geq \frac{1}{2K}(y-x)^2,
\end{equation*}
which gives the desired result.
\end{proof}

It is a well-known result that the free energy $F$ is nonnegative \cite{ge2010physical}. The following theorem strengthens this result by giving a nonzero lower bound of the free energy $F$.
\begin{theorem}\label{lowerF}
The following inequality holds:
\begin{equation*}
F \geq \frac{1}{2}\sum_i(p_i-\mu_i)^2.
\end{equation*}
\end{theorem}

\begin{proof}
Since $0\leq\mu_i,p_i\leq 1$, it follows from Lemma \ref{basic} that
\begin{equation*}
p_i\log\frac{p_i}{\mu_i} \geq \frac{1}{2}(p_i-\mu_i)^2+(p_i-\mu_i).
\end{equation*}
By the definition of $F$ in \eqref{F} we have
\begin{equation*}
F \geq \sum_i\left[\frac{1}{2}(p_i-\mu_i)^2+(p_i-\mu_i)\right] = \frac{1}{2}\sum_i(p_i-\mu_i)^2,
\end{equation*}
which gives the desired result.
\end{proof}

The next theorem gives a nonzero lower bound for the non-adiabatic entropy production rate $e_p^{(na)}$, also known as the free energy dissipation rate $f_d = -dF/dt$.
\begin{theorem}\label{lowerfd}
There exists a constant $c_1>0$ such that
\begin{equation*}
e_p^{(na)} = f_d = -\frac{dF}{dt} \geq c_1\sum_i\left[\sum_{j\neq i}(p_jq_{ji}-p_iq_{ij})\right]^2,
\end{equation*}
where $c_1$ only depends on the transition rate matrix $Q$.
\end{theorem}

\begin{proof}
Since $0<\mu_i,p_i\leq 1$, it follows from Lemma \ref{basic} that
\begin{equation*}
p_i\mu_j\log\frac{p_i\mu_j}{p_j\mu_i}
\geq \frac{1}{2}(p_i\mu_j-p_j\mu_i)^2+(p_i\mu_j-p_j\mu_i).
\end{equation*}
By the definition of $f_d$ in \eqref{fd}, we have
\begin{equation*}
\begin{split}
e_p^{(na)} &= \sum_{i,j}p_iq_{ij}\log\frac{p_i\mu_j}{p_j\mu_i}
= \sum_{i,j}\frac{q_{ij}}{\mu_j}p_i\mu_j\log\frac{p_i\mu_j}{p_j\mu_i} \\
&\geq \sum_{i,j}\frac{q_{ij}}{\mu_j}\left[\frac{1}{2}(p_i\mu_j-p_j\mu_i)^2+(p_i\mu_j-p_j\mu_i)\right] \\
&= \frac{1}{2}\sum_{i,j}\frac{q_{ij}}{\mu_j}(p_i\mu_j-p_j\mu_i)^2+\sum_{i,j}p_iq_{ij}
-\sum_{i,j}\frac{p_j}{\mu_j}\mu_iq_{ij} \\
&= \frac{1}{2}\sum_{i,j}\frac{q_{ij}}{\mu_j}(p_i\mu_j-p_j\mu_i)^2,
\end{split}
\end{equation*}
where the last equality is due to $\sum_jq_{ij} = 0$ and $\sum_i\mu_iq_{ij} = 0$. Set $M = \max\{q_{ij}:i\neq j\}>0$. Then we have
\begin{equation*}
e_p^{(na)} \geq \frac{1}{2}\sum_{i,j}\frac{q_{ij}}{\mu_j}\frac{q_{ij}}{M}(p_i\mu_j-p_j\mu_i)^2 = \frac{1}{2M}\sum_j\mu_j\sum_i\left[\frac{q_{ij}}{\mu_j}(p_i\mu_j-p_j\mu_i)\right]^2.
\end{equation*}
Recall that the components of $\mu$ are all positive. Set $r = \min\{\mu_1,\mu_2\cdots,\mu_N\}>0$. It thus follows from the Cauchy-Schwarz inequality that
\begin{equation*}
\begin{split}
e_p^{(na)} &\geq \frac{r}{2MN}\sum_j\left[\sum_i\frac{q_{ij}}{\mu_j}(p_i\mu_j-p_j\mu_i)\right]^2
= \frac{r}{2MN}\sum_j\left[\sum_ip_iq_{ij}\right]^2 \\
&= \frac{r}{2MN}\sum_j\left[\sum_{i\neq j}(p_iq_{ij}-p_jq_{ji})\right]^2,
\end{split}
\end{equation*}
which gives the desired result.
\end{proof}

\begin{remark}
This theorem is inspired by the entropy-dissipation principle in nonequilibrium thermodynamics proposed in \cite{yong2004entropy} and has been proved when the steady-steady distribution $\mu$ satisfies the detailed balance condition \cite{yong2008interesting, yong2012conservation}. However, we do not need detailed balance here and thus the above theorem can be applied to general Markov jump processes.
\end{remark}

The following result is a direct corollary of Theorems \ref{lowerF} and \ref{lowerfd}.
\begin{corollary}
The following three statements are equivalent:\\
(a) The system is in a steady state; \\
(b) The free energy $F$ vanishes; \\
(c) The free energy dissipation rate $f_d$ vanishes.
\end{corollary}

\begin{proof}
If (a) holds, then $p = \mu$ and thus (b) and (c) follow from \eqref{F} and \eqref{fd} immediately. If (b) holds, then it follows from Theorem \ref{lowerF} that $p = \mu$, which shows that (a) holds. If (c) holds, then it follows from Theorem \ref{lowerfd} that for any $i$,
\begin{equation*}
\sum_{j\neq i}p_iq_{ij} = \sum_{j\neq i}p_jq_{ji}.
\end{equation*}
This shows that $p$ is a steady-state distribution and thus (a) holds.
\end{proof}

\begin{remark}
This corollary shows that the free energy dissipation rate vanishes if and only if the system is in a steady state. Therefore, a nonequilibrium system will never stop dissipating free energy unless it has reached the steady state. Thus the free energy dissipation rate characterizes the irreversibility in the spontaneous relaxation process towards the steady state. This is an interesting echo of the famous H-theorem for the Boltzmann equation, which characterizes the relaxation process of a thermodynamic system towards thermodynamic equilibrium.
\end{remark}

\subsection{Nonzero lower bound for the adiabatic entropy production rate}
We have given a nonzero lower bound for the non-adiabatic entropy production rate. The following theorem provides an analogue for the adiabatic part.
\begin{theorem}\label{lowerepad1}
There exists a constant $c_2>0$ such that
\begin{equation*}
e_p^{(ad)} \geq c_2\sum_{i,j}p_i(\mu_iq_{ij}-\mu_jq_{ji})^2,
\end{equation*}
where $c_2$ only depends on the transition rate matrix $Q$.
\end{theorem}

\begin{proof}
Set $M = \max\{q_{ij}:i\neq j\}>0$. Since $0\leq\mu_iq_{ij}\leq M$ for any $i\neq j$, it follows from Lemma \ref{basic} that
\begin{equation*}
\mu_iq_{ij}\log\frac{\mu_iq_{ij}}{\mu_jq_{ji}}
\geq \frac{1}{2M}(\mu_iq_{ij}-\mu_jq_{ji})^2+(\mu_iq_{ij}-\mu_jq_{ji}).
\end{equation*}
By the definition of $e_p^{(ad)}$ in \eqref{epad}, we have
\begin{equation*}
\begin{split}
e_p^{(ad)} &= \sum_{i,j}p_iq_{ij}\log\frac{\mu_iq_{ij}}{\mu_jq_{ji}}
= \sum_{i,j}\frac{p_i}{\mu_i}\mu_iq_{ij}\log\frac{\mu_iq_{ij}}{\mu_jq_{ji}} \\
&\geq \sum_{i,j}\frac{p_i}{\mu_i}
\left[\frac{1}{2M}(\mu_iq_{ij}-\mu_jq_{ji})^2+(\mu_iq_{ij}-\mu_jq_{ji})\right] \\
&= \frac{1}{2M}\sum_{i,j}\frac{p_i}{\mu_i}(\mu_iq_{ij}-\mu_jq_{ji})^2
\geq \frac{1}{2M}\sum_{i,j}p_i(\mu_iq_{ij}-\mu_jq_{ji})^2,
\end{split}
\end{equation*}
which gives the desired result.
\end{proof}

The following result is a direct corollary of the above theorem.
\begin{corollary}\label{lowerepad2}
Assume that the components of the initial distribution $p(0)$ are all positive. Then there exists a constant $c_3>0$ such that
\begin{equation*}
e_p^{(ad)} \geq c_3\sum_{i,j}(\mu_iq_{ij}-\mu_jq_{ji})^2,
\end{equation*}
where $c_3$ depends on the transition rate matrix $Q$ and the initial distribution $p(0)$.
\end{corollary}

\begin{proof}
By Theorem \ref{lowerepad1}, there exists a constant $c_2>0$ depending on the transition rate matrix $Q$ such that
\begin{equation}\label{general}
e_p^{(ad)} \geq c_2\sum_{i,j}p_i(\mu_iq_{ij}-\mu_jq_{ji})^2.
\end{equation}
Let $r$ be a constant defined by
\begin{equation*}
r = \min_i\min_{t\geq 0}p_i(t).
\end{equation*}
Since the components of $p(0)$ are all positive, it is easy to see  that the components of $p(t)$ are all positive for any $t\geq 0$. Since the system is irreducible, the components of $\mu$ are all positive and $p(t)\rightarrow\mu$ as $t\rightarrow\infty$ \cite{norris1998markov}. The above two facts imply that $r>0$. It thus follows from \eqref{general} that
\begin{equation*}
e_p^{(ad)} \geq c_2r\sum_{i,j}(\mu_iq_{ij}-\mu_jq_{ji})^2,
\end{equation*}
which gives the desired result.
\end{proof}

\begin{remark}
We know from \cite{norris1998markov} that even if some components of the initial distribution $p(0)$ are zero, all components will become positive after an arbitrarily small time. Therefore, the assumption of the above corollary does not impose too much restriction on the system. In addition, it has been shown that if the steady state of the system is an NESS, then the adiabatic entropy production rate $e_p^{(ad)}>0$ \cite{ge2010physical, esposito2010three}. The above corollary further strengthens this result by providing a positive lower bound for $e_p^{(ad)}$, which is also time-independent. This reveals an intrinsic strong dissipative property for the master equation.
\end{remark}

The next result follows directly from Corollary \ref{lowerepad2}.
\begin{corollary}
Assume that the components of the initial distribution $p(0)$ are all positive. Then the following three statements are equivalent: \\
(a) The steady state of the system is an equilibrium state; \\
(b) The steady state of the system satisfies the detailed balance condition; \\
(c) The adiabatic entropy production rate $e_p^{(ad)}$ vanishes.
\end{corollary}

\begin{proof}
The equivalence of (a) and (b) follows from the definition of the equilibrium state. We next prove the equivalence of (b) and (c). If (b) is true, then the detailed balance condition $\mu_iq_{ij} = \mu_jq_{ji}$ holds for any pair of states $i$ and $j$. It thus follows from \eqref{epad} that $e_p^{(ad)}=0$. On the other hand, if (c) holds, then it follows from Theorem \ref{lowerepad1} that
\begin{equation*}
\sum_{i,j}p_i(\mu_iq_{ij}-\mu_jq_{ji})^2 = 0.
\end{equation*}
Since the components of the initial distribution $p(0)$ are all positive, the components of $p(t)$ are all positive for any $t\geq 0$. This shows that $\mu_iq_{ij} = \mu_jq_{ji}$ for any pair of states $i$ and $j$, which implies that (b) holds.
\end{proof}

\begin{remark}
This corollary shows that adiabatic entropy production rate vanishes if and only if the steady state of the system is an equilibrium state. Thus the adiabatic part reflects the irreversibility in an NESS caused by the breaking of detailed balance. In fact, even if the system has reaches the steady state, some kind of circular motions may still exist to maintain an NESS and give rise to a positive $e_p^{(ad)}$ \cite{jiang2004mathematical}. This phenomenon is typical in many biochemical processes and constitutes a major difference between an NESS and an equilibrium state \cite{zhang2012stochastic}.
\end{remark}

\subsection{Nonzero lower bound for the total entropy production rate}
We have given nonzero lower bounds for the adiabatic and non-adiabatic entropy production rates, which automatically give rise to a nonzero lower bound for the total entropy production rate $e_p$. The following theorem provides another nonzero lower bound for $e_p$.
\begin{theorem}\label{lowerep}
There exists constants $c_4>0$ such that
\begin{equation*}
e_p \geq c_4\sum_{i,j}(p_iq_{ij}-p_jq_{ji})^2,
\end{equation*}
where $c_4$ only depends on the transition rate matrix $Q$.
\end{theorem}

\begin{proof}
Set $M = \max\{q_{ij}:i\neq j\}>0$. Since $0\leq p_iq_{ij}\leq M$ for any $i\neq j$, it follows from Lemma \ref{basic} that
\begin{equation*}
p_iq_{ij}\log\frac{p_iq_{ij}}{p_jq_{ji}}
\geq \frac{1}{2M}(p_iq_{ij}-p_jq_{ji})^2+(p_iq_{ij}-p_jq_{ji}).
\end{equation*}
By the definition of $e_p$ in \eqref{ep}, we have
\begin{equation*}
\begin{split}
e_p &= \sum_{i,j}p_iq_{ij}\log\frac{p_iq_{ij}}{p_jq_{ji}}
\geq \sum_{i,j}\left[\frac{1}{2M}(p_iq_{ij}-p_jq_{ji})^2+(p_iq_{ij}-p_jq_{ji})\right] \\
&= \frac{1}{2M}\sum_{i,j}(p_iq_{ij}-p_jq_{ji})^2+\sum_{i,j}(p_iq_{ij}-p_jq_{ji}) = \frac{1}{2M}\sum_{i,j}(p_iq_{ij}-p_jq_{ji})^2,
\end{split}
\end{equation*}
which gives the desired result.
\end{proof}

\begin{remark}
It is interesting to note that, in the above theorem, the lower bound for the entropy production rate $e_p$ is given by the sum of squares of the local fluxes $J_{ij} = p_iq_{ij}-p_jq_{ji}$ between each pair of states $i$ and $j$. In contrast, it is shown in Theorem \ref{lowerfd} that the free energy dissipation rate $f_d$ is bounded from below by the sum of squares of the total fluxes $J_i = \sum_{j\neq i}(p_jq_{ji}-p_iq_{ij})$ through each state $i$.
\end{remark}

\subsection{Results for the Tsallis-type relative entropy}
Theorems \ref{lowerF} and \ref{lowerfd} give the nonzero lower bounds for the free energy and its dissipation rate, where the free energy can be expressed as the Boltzmann-Shannon relative entropy between $p$ and $\mu$. Here we try to extend these results to the Tsallis-type relative entropy, which was introduced by Tsallis in \cite{tsallis1988possible} as a generalization of the classical Boltzmann-Gibbs statistics. The resulting theory is considered to be important for the non-extensive thermodynamics. It has also been found applications in a wide range of natural, artificial, and social complex systems \cite{tsallis2011nonadditive}. Specifically, we recall the following definition.

\begin{definition}
Let $u = (u_1,u_2,\cdots,u_N)$ and $v = (v_1,v_2,\cdots,v_N)$ be two probability distributions where the components of $v$ are all positive. Then, for any real number $\alpha\neq 0,1$, the Tsallis-type relative entropy of order $\alpha$ between $u$ and $v$ is defined as
\begin{equation*}
D_{\alpha}(u\|v) = \frac{1}{\alpha(\alpha-1)}\left[\sum_iu_i\left(\frac{u_i}{v_i}\right)^{\alpha-1}-1\right].
\end{equation*}
\end{definition}

By L'Hospital's rule, it is easy to see that
\begin{equation*}
\lim_{\alpha\rightarrow1}D_{\alpha}(u\|v) = \lim_{\alpha\rightarrow1}\sum_iu_i\left(\frac{u_i}{v_i}\right)^{\alpha-1}\log\frac{u_i}{v_i} = D(u\|v).
\end{equation*}
This shows that the Tsallis-type relative entropy converges to the Boltzmann-Shannon one as $\alpha\rightarrow1$. In analogy to the definition of the free energy $F$, we define the Tsallis-type free energy $F_\alpha$ as
\begin{equation*}
F_\alpha = D_{\alpha}(p\|\mu) = \frac{1}{\alpha(\alpha-1)}\left[\sum_ip_i\left(\frac{p_i}{\mu_i}\right)^{\alpha-1}-1\right].
\end{equation*}
It has been proved in \cite{shiino1998h} that $F_\alpha$ has the following nice properties similar to those of the free energy $F$:
\begin{equation*}
F_{\alpha} \geq 0,\;\;\;-\frac{dF_{\alpha}}{dt} \geq 0.
\end{equation*}
Next we shall strengthen these results by providing nonzero lower bounds for the Tsallis-type free energy $F_\alpha$ and its dissipation rate $-dF_\alpha/dt$.

To this end, we need the following two lemmas.
\begin{lemma}\label{MP} (See \cite{mitrinovic2013classical})
Let $\alpha$ and $\beta$ be two real numbers satisfying $\alpha<2,\alpha\neq 0,1$ and $\beta\geq 0$. Then for any $x\in[-1,\beta]$, the following inequality holds:
\begin{equation*}
\frac{1}{\alpha(\alpha-1)}[(x+1)^{\alpha}-\alpha x-1] \geq \frac{1}{2}(\beta+1)^{\alpha-2}x^2,
\end{equation*}
where the case of $x=-1$ is understood in the limit sense.
\end{lemma}

\begin{lemma}\label{Leindler} (See \cite{leindler1972on})
For any $\alpha\geq 2$, there exists a constant $u_\alpha>0$ depending on $\alpha$ such that for any complex number $z$, the following inequality holds:
\begin{equation*}
\frac{1}{\alpha(\alpha-1)}[|z+1|^{\alpha}-1-\alpha\textrm{Re}(z)] \geq u_\alpha|z|^\alpha.
\end{equation*}
\end{lemma}

The following theorem gives a nonzero lower bound for the Tsallis-type free energy.
\begin{theorem}\label{lowerfalpha}
For any $\alpha\neq 0,1$,
\begin{equation*}
F_\alpha \geq
\begin{cases}
\frac{1}{2}\sum\limits_i\mu_i^{1-\alpha}(p_i-\mu_i)^2, &\textrm{if\;} \alpha<2, \\
u_\alpha \sum\limits_i\mu_i^{1-\alpha}|p_i-\mu_i|^{\alpha}, &\textrm{if\;} \alpha\geq 2,
\end{cases}
\end{equation*}
where $u_\alpha$ is the constant in Lemma \ref{Leindler}.
\end{theorem}

\begin{proof}
Let $x_i = p_i/\mu_i-1$ and $\beta = 1/\mu_i-1\geq 0$. Then $x_i\in[-1,\beta]$ and it follows from Lemma \ref{MP} that for any $\alpha<2$ and $\alpha\neq 0,1$,
\begin{equation*}
\frac{1}{\alpha(\alpha-1)}[(x_i+1)^\alpha-\alpha x_i-1] \geq \frac{1}{2}\mu_i^{2-\alpha}x_i^2 = \frac{1}{2}\mu_i^{-\alpha}(p_i-\mu_i)^2.
\end{equation*}
Thus we have
\begin{equation*}
\begin{split}
F_\alpha &= \frac{1}{\alpha(\alpha-1)}\left[\sum_i\mu_i(x_i+1)^\alpha-1\right]
= \frac{1}{\alpha(\alpha-1)}\sum_i\mu_i[(x_i+1)^\alpha-\alpha x_i-1] \\
&\geq \frac{1}{2}\sum_i\mu_i^{1-\alpha}(p_i-\mu_i)^2.
\end{split}
\end{equation*}
This shows that the theorem holds when $\alpha<2$. On the other hand, it follows from Lemma \ref{Leindler} that for any $\alpha\geq 2$,
\begin{equation*}
\begin{split}
F_\alpha &= \frac{1}{\alpha(\alpha-1)}\left[\sum_i\mu_i(x_i+1)^\alpha-1\right]
= \frac{1}{\alpha(\alpha-1)}\sum_i\mu_i[(x_i+1)^\alpha-\alpha x_i-1] \\
&\geq u_\alpha\sum_i\mu_i|x_i|^\alpha = u_\alpha\sum_i\mu_i^{1-\alpha}|p_i-\mu_i|^\alpha.
\end{split}
\end{equation*}
This shows that the theorem also holds when $\alpha\geq 2$.
\end{proof}

The following theorem gives a nonzero lower bound for the dissipation rate $-dF_\alpha/dt$ of the Tsallis-type free energy.
\begin{theorem}\label{lowerfdalpha}
For any $\alpha\neq 0,1$, the Tsallis-type free energy has the following dissipative property:
\begin{equation*}
-\frac{dF_{\alpha}}{dt} \geq
\begin{cases}
c_\alpha \sum\limits_i[\sum\limits_{j\neq i}(q_{ij}p_j-q_{ji}p_i)]^2, &\textrm{if\;} \alpha<2, \\
c_\alpha \sum\limits_i|\sum\limits_{j\neq i}(q_{ij}p_j-q_{ji}p_i)|^\alpha, &\textrm{if\;} \alpha\geq 2,
\end{cases}
\end{equation*}
where $c_\alpha>0$ is a constant only depending on $\alpha$ and the transition rate matrix $Q$.
\end{theorem}

\begin{proof}
Let $y_i = p_i/\mu_i$. It follows from the master equation \eqref{master} that
\begin{equation*}
\begin{split}
-\frac{dF_{\alpha}}{dt} &= -\frac{1}{\alpha-1}\sum_i\mu_iy_i^{\alpha-1}\frac{dy_i}{dt}
= -\frac{1}{\alpha-1}\sum_iy_i^{\alpha-1}\frac{dp_i}{dt} \\
&= \frac{1}{\alpha-1}\sum_iy_i^{\alpha-1}\sum_{j\neq i}(p_iq_{ij}-p_jq_{ji})
= \frac{1}{\alpha-1}\sum_{i,j\neq i}p_iq_{ij}(y_i^{\alpha-1}-y_j^{\alpha-1}) \\
&= \frac{1}{\alpha-1}\sum_{i,j\neq i}\mu_iq_{ij}(y_i^\alpha-y_iy_j^{\alpha-1})+\frac{1}{\alpha}\sum_{i,j}\mu_iq_{ij}(y_j^\alpha-y_i^\alpha) \\
&= \frac{1}{\alpha(\alpha-1)}\sum_{i,j\neq i}\mu_iq_{ij}[y_i^\alpha-\alpha y_iy_j^{\alpha-1}+(\alpha-1)y_j^\alpha].
\end{split}
\end{equation*}
Let $z_{ij} = y_i/y_j-1$. It is easy to see that
\begin{equation}\label{temp}
-\frac{dF_{\alpha}}{dt} = \frac{1}{\alpha(\alpha-1)}\sum_{i,j\neq i}\mu_iq_{ij}y_j^\alpha[(z_{ij}+1)^\alpha-\alpha z_{ij}-1].
\end{equation}
Let $\beta = 1/\mu_ip_j-1\geq 0$. Then $z_{ij}\in[-1,\beta]$ and it follows from \eqref{temp} and Lemma \ref{MP} that for any $\alpha<2$ and $\alpha\neq 0,1$,
\begin{equation*}
\begin{split}
-\frac{dF_{\alpha}}{dt} &\geq \frac{1}{2}\sum_{i,j\neq i}\mu_iq_{ij}y_j^\alpha(\mu_ip_j)^{2-\alpha}z_{ij}^2
= \frac{1}{2}\sum_{i,j\neq i}\mu_iq_{ij}y_j^{\alpha-2}(\mu_ip_j)^{2-\alpha}(y_i-y_j)^2 \\
&= \frac{1}{2}\sum_{i,j}\mu_iq_{ij}(\mu_i\mu_j)^{2-\alpha}(y_i-y_j)^2.
\end{split}
\end{equation*}
Recall that the components of $\mu$ are all positive. Set $r = \min\{\mu_1,\mu_2\cdots,\mu_N\}>0$ and $M = \max\{q_{ij}:i\neq j\}>0$. It thus follows from the Cauchy-Schwarz inequality that
\begin{equation*}
\begin{split}
-\frac{dF_{\alpha}}{dt} &\geq \frac{r^{4-2\alpha}}{2}\sum_{i,j}\mu_iq_{ij}(y_i-y_j)^2
\geq \frac{r^{4-2\alpha}}{2}\sum_{i,j}\mu_iq_{ij}\frac{\mu_iq_{ij}}{M}(y_i-y_j)^2 \\
&\geq \frac{r^{4-2\alpha}}{2MN}\sum_j\left[\sum_i\mu_iq_{ij}(y_i-y_j)\right]^2
= \frac{r^{4-2\alpha}}{2MN}\sum_j\left[\sum_ip_iq_{ij}\right]^2 \\
&= \frac{r^{4-2\alpha}}{2MN}\sum_j\left[\sum_{i\neq j}(p_iq_{ij}-p_jq_{ji})\right]^2.
\end{split}
\end{equation*}
This shows that the theorem holds when $\alpha<2$. On the other hand, it follows from \eqref{temp} and Lemma \ref{Leindler} that for any $\alpha\geq 2$,
\begin{equation*}
\begin{split}
-\frac{dF_{\alpha}}{dt} &\geq u_\alpha\sum_{i,j\neq i}\mu_iq_{ij}y_j^\alpha|z_{ij}|^\alpha
= u_\alpha\sum_{i,j}\mu_iq_{ij}|y_i-y_j|^\alpha \\
&\geq u_\alpha\sum_{i,j}\mu_iq_{ij}\left(\frac{\mu_iq_{ij}}{M}\right)^{\alpha-1}|y_i-y_j|^\alpha \\
&= \frac{u_\alpha}{M^{\alpha-1}}\sum_{i,j}(\mu_iq_{ij}|y_i-y_j|)^\alpha.
\end{split}
\end{equation*}
It thus follows from the H\"{o}lder inequality that
\begin{equation*}
\begin{split}
-\frac{dF_{\alpha}}{dt} &\geq \frac{u_\alpha}{(MN)^{\alpha-1}}\sum_j\left[\sum_i\mu_iq_{ij}|y_i-y_j|\right]^\alpha
\geq \frac{u_\alpha}{(MN)^{\alpha-1}}\sum_j\left|\sum_i\mu_iq_{ij}(y_i-y_j)\right|^\alpha \\
&= \frac{u_\alpha}{(MN)^{\alpha-1}}\sum_j\left|\sum_ip_iq_{ij}\right|^\alpha
= \frac{u_\alpha}{(MN)^{\alpha-1}}\sum_j\left|\sum_{i\neq j}(p_iq_{ij}-p_jq_{ji})\right|^\alpha.
\end{split}
\end{equation*}
This shows that the theorem also holds when $\alpha\geq 2$.
\end{proof}

We conclude this section with the following observation.

\begin{remark}
It is interesting to note that by taking $\alpha\rightarrow 1$ in Theorem \ref{lowerfalpha}, we obtain that
\begin{equation*}
F = \lim_{\alpha\rightarrow 1}F_\alpha \geq \frac{1}{2}\sum_i(p_i-\mu_i)^2.
\end{equation*}
This is exactly the inequality stated in Theorem \ref{lowerF} which gives the nonzero lower bound for the free energy $F$. Similarly, by taking $\alpha\rightarrow 1$ formally in Theorem \ref{lowerfdalpha}, we can recover the nonzero lower bound for the free energy dissipation rate $f_d$ proved in Theorem \ref{lowerfd}.
\end{remark}

\section{Physical significance of the nonzero lower bounds}
In this section, we shall present some physical interpretations of the nonzero lower bounds obtained in the previous section. It is well-known that the entropy production rate of a system is always nonnegative and this is actually the second law of thermodynamics, which characterizes the irreversibility of nonequilibrium processes. Recently, it has been shown in \cite{ge2010physical, esposito2010three} that the irreversibility has two different mechanisms: the deviation from steady state and the breaking of detailed balance. The former is quantified by the non-adiabatic entropy production rate $e_p^{(na)}$, while the latter by the adiabatic entropy production rate $e_p^{(ad)}$. Therefore, the decomposition of the total entropy production rate $e_p = e_p^{(na)}+e_p^{(ad)}$ provides a deep insight into the irreversibility. This important point of view is by no means obvious without the aid of the nonzero lower bounds obtained above.

For a long time, it is known that the entropy production rate $e_p$ of a system can generally be written as a bilinear form of the thermodynamic fluxes and forces \cite{de1962non}. For the master equation, we have
\begin{equation*}
e_p = \sum_{i,j}J_{ij}X_{ij} = \sum_{i<j}(p_iq_{ij}-p_jq_{ji})\log\frac{p_iq_{ij}}{p_jq_{ji}},
\end{equation*}
where $J_{ij} = p_iq_{ij}-p_jq_{ji}$ is the local thermodynamic flux between states $i$ and $j$ and
\begin{equation*}
X_{ij} = \log\frac{p_iq_{ij}}{p_jq_{ji}}
\end{equation*}
is the corresponding local thermodynamic force. In fact, the non-adiabatic part $e_p^{(na)}$ can be written in a similar way as
\begin{equation*}
e_p^{(na)} = -\frac{dF}{dt} = -\sum_i\frac{\partial F}{\partial p_i}\frac{dp_i}{dt} = \sum_iJ_iX_i,
\end{equation*}
where $J_i = dp_i/dt = \sum_{j\neq i}(p_jq_{ji}-p_iq_{ij})$ is the total thermodynamic flux through state $i$ and
\begin{equation*}
X_i = -\frac{\partial F}{\partial p_i} = -\left[\log\frac{p_i}{\mu_i}+1\right]
\end{equation*}
is the corresponding total thermodynamic force. Thus both the entropy production rate $e_p$ and its non-adiabatic part $e_p^{(na)}$ can be cast into a similar bilinear form. The former concerns with local thermodynamic fluxes and forces and the latter only deals with the total ones.

With Theorems \ref{lowerep} and \ref{lowerfd}, we show that there exist constants $c_1,c_4>0$ only depending on the transition rate matrix $Q$ such that
\begin{equation*}
e_p = \sum_{i,j}J_{ij}X_{ij} \geq c_4\sum_{i,j}J_{ij}^2
\end{equation*}
and
\begin{equation*}
e_p^{(na)} = \sum_iJ_iX_i \geq c_1\sum_iJ_i^2.
\end{equation*}
These inequalities reveal two strong dissipative properties for the master equation. They suggest that a bilinear form of the thermodynamic fluxes and forces generally has a nonzero lower bound which is proportional to the sum of squares of the thermodynamic fluxes. We believe that this property is possessed by quite general thermodynamic systems \cite{yong2008interesting}, not necessarily modeled by the master equation.

Recall that a central task of modern nonequilibrium thermodynamics is to provide guiding principles for  modeling various irreversible processes, where the second law of thermodynamics plays a fundamental role. For example, classical irreversible thermodynamics (CIT), rational extended thermodynamics (RET), and extended irreversible thermodynamics (EIT) are all such theories that have been developed and applied with great success for this task \cite{de1962non, muller1998rational, jou2010extended}. Therefore, the strong dissipative properties obtained above may provide such a guiding principle for the construction of thermodynamic constitutive relations.

Let $J = (J_k)$ and $X = (X_k)$ denote two column vectors composed of all thermodynamic fluxes and forces of a molecular system, respectively. In nonequilibrium thermodynamics, we hope to find a constitute relation $J = J(X)$ between the thermodynamic fluxes and forces. In this way, we can obtain a closed and solvable mathematical model together with balance equations \cite{zhu2014conservation}. In the region not far from equilibrium, the theory of CIT claims a linear constitutive relation: $J(X) = MX$, where $M$ is a constant matrix which should be nonnegative definite in order to guarantee $e_p\geq 0$. The celebrated Onsager's reciprocal relation further requires $M$ to be symmetric \cite{reichl1980modern}.

However, when the system is far from equilibrium, the matrix $M$ cannot be viewed as constant. A possible modification is the nonlinear constitutive relation of $J(X) = M(X)X$, where $M(X)$ is a nonnegative definite matrix depending on $X$. This kind of relations is allowed by the second law of thermodynamics. In practice, we are often concerned with the asymptotic behavior of $J(X)$ as $X$ tends to zero or infinity. To this end, we assume that $J(X)$ behaves like $|X|^\alpha$ as $X$ tends to zero or infinity, namely, there exist constants $\gamma_1>0$ and $\gamma_2>0$ such that
\begin{equation}\label{elliptic}
\gamma_1|X|^\alpha \leq |J(X)| \leq \gamma_2|X|^\alpha.
\end{equation}

Remarkably, the nonzero lower bounds obtained above impose additional constraints on the exponent $\alpha$. In fact, it is easy to see from \eqref{elliptic} that
\begin{equation}\label{right}
e_p = J\cdot X = J(X)\cdot X \leq |X|\cdot|J(X)| \leq \gamma_2|X|^{\alpha+1}.
\end{equation}
On the other hand, if the entropy production rate $e_p$ has a nonzero lower bound as in Theorem \ref{lowerfd}, it is easy to check that
\begin{equation}\label{left}
e_p \geq c_4|J|^2 = c_4|J(X)|^2 \geq c_4\gamma_1^2|X|^{2\alpha}.
\end{equation}
Combining \eqref{right} and \eqref{left}, we obtain that
\begin{equation}\label{left}
|X|^{\alpha-1} \leq \frac{\gamma_2}{c_4\gamma_1^2}.
\end{equation}
This implies that $\alpha$ must satisfy $\alpha\geq 1$ as $X\rightarrow 0$ and $\alpha\leq 1$ as $X\rightarrow\infty$. Therefore, the nonzero lower bounds obtained above indeed provide guidance to the construction of thermodynamic constitutive relations and such guidance does not seem to be reported in the literature before.

\section{A mathematical application of the nonzero lower bounds}
The nonzero lower bounds obtained in this paper not only provide physical insights into the dissipative properties of nonequilibrium processes, but also imply some mathematical consequences such as the long-time dynamics of the master equation.

It is a classical result that the probability distribution $p(t)$ of an irreducible Markov jump process will converge to the steady-state distribution $\mu$ as $t\rightarrow\infty$. This fact can be proved either by the algebraic method based on the Perron-Frobenius theorem \cite{berman1979nonnegative} or by the probabilistic method based on the coupling of Markov chains \cite{norris1998markov}. However, both the two methods turn out to be rather involved. In this section, we shall give a simple and purely analytic proof of this classical result by using the nonzero lower bounds obtained in this paper.

To this end, we recall the following elementary fact which will be proved for completeness.
\begin{lemma}\label{long}
Let $f$ be a Lipschitz continuous and integrable function on $[0,\infty)$. Then we have 
\begin{equation}
\lim_{t\rightarrow\infty}f(t) = 0.
\end{equation}
\end{lemma}

\begin{proof}
Since $f$ is Lipschitz continuous, there exists a constant $K>0$ such that $|f(x)-f(y)|<K|x-y|$ for any $x,y\geq 0$. Assume that $f(t)$ does not converge to $0$ as $t\rightarrow\infty$. Then there exist $0<\epsilon<K$ and a sequence $t_n\rightarrow\infty$ such that $t_{n+1}-t_n>1$ and $|f(t_n)|>2\epsilon$ for any $n\geq 1$. For any $n\geq 1$ and $h<\epsilon/K$, it is easy to see that
\begin{equation*}
|f(t_n+h)-f(t_n)| \leq Kh < \epsilon.
\end{equation*}
This shows that $|f(t_n+h)|>\epsilon$ for any $n\geq 1$ and $h<\epsilon/K$. Thus we obtain that
\begin{equation*}
\int_0^\infty|f(t)|dt \geq \sum_{n=1}^\infty\int_{t_n}^{t_n+\epsilon/K}|f(t)|dt \geq \sum_{n=1}^\infty\frac{\epsilon^2}{K} = \infty.
\end{equation*}
This contradicts the integrability of $f$ on $[0,\infty)$ and hence we obtain the desired result.
\end{proof}

We are now in a position to study the long-time dynamics of Markov jump processes based on the nonzero lower bound for the free energy dissipation rate.
\begin{theorem}
Assume that the system is irreducible. Then we have 
\begin{equation*}
\lim_{t\rightarrow\infty}p(t) = \mu.
\end{equation*}
\end{theorem}

\begin{proof}
Define a function $f=f(t)$ for $t\in[0,\infty)$ as
\begin{equation*}
f(t) = \sum_i\left[\sum_{j\neq i}(p_j(t)q_{ji}-p_i(t)q_{ij})\right]^2 = \sum_i\left[\sum_jp_j(t)q_{ji}\right]^2.
\end{equation*}
Integrating the inequality in Theorem \ref{lowerfd} and noting that $F(t)\geq 0$, we obtain that
\begin{equation*}
c_1\int_0^tf(s)ds \leq F(0)-F(t)\leq F(0).
\end{equation*}
This indicates that $f$ is integrable on $[0, \infty)$.

Moreover, set $M = \max\{|q_{ij}|:1\leq i,j\leq N\}$. It follows from the master equation \eqref{master} that for any $i$,
\begin{equation}\label{temp1}
|\dot{p_i}(t)| \leq \sum_j|p_j(t)q_{ij}| \leq MN.
\end{equation}
Direct computation shows that
\begin{equation}\label{temp2}
\dot{f}(t) = 2\sum_i\left[\sum_jp_j(t)q_{ji}\right]\left[\sum_j\dot{p}_j(t)q_{ji}\right].
\end{equation}
From \eqref{temp1} and \eqref{temp2}, it is easy to see that $|\dot{f}(t)|$ is uniformly bounded and thus $f$ is Lipschitz continuous on $[0, \infty)$. It thus follows from Theorem \ref{long} that
\begin{equation}\label{zero}
\lim_{t\rightarrow\infty}f(t) = 0.
\end{equation}

Assume that $p(t)$ does not converge to $\mu$ as $t\rightarrow\infty$. Since $p(t)$ is uniformly bounded, there exist a state $\mu'$ and a sequence $t_n\rightarrow\infty$ such that
\begin{equation*}
\lim_{n\rightarrow\infty}p(t_n)= \mu' \neq \mu.
\end{equation*}
Owing to the uniqueness of the steady-state distribution $\mu$, it is obvious that
\begin{equation*}
r := \sum_i\left[\sum_j\mu'_jq_{ji}\right]^2 > 0.
\end{equation*}
Thus when $n$ is sufficiently large, we have
\begin{equation*}
f(t_n) = \sum_i\left[\sum_jp_j(t_n)q_{ji}\right]^2 \geq \frac{r}{2}.
\end{equation*}
This contradicts \eqref{zero} and hence we obtain the desired result.
\end{proof}

\section*{Acknowledgment}
This work is supported by the National Natural Science Foundation of China (Grants 11204150 and 11471185) and the Initiative Scientific Research Program of Tsinghua University (Grant 20151080424).

\setlength{\bibsep}{5pt}
\small\bibliographystyle{pnas2009}
\bibliography{master}

\end{document}